\documentclass[11pt]{article}

\usepackage[margin=1in]{geometry}
\usepackage{graphicx}
\usepackage[pdftex,colorlinks=true,linkcolor=blue,citecolor=blue,urlcolor=black]{hyperref}
\usepackage{amsmath, amsthm, amssymb}
\usepackage{subfigure}
\usepackage{comment}
\usepackage{url}
\usepackage{pdflscape}
\usepackage[ruled,lined,linesnumbered]{algorithm2e}

\def \H{{\cal H}}
\newcommand{\R}{\mathbb{R}}

\newcommand{\Z}{\mathbb{Z}}

\newcommand{\ket}[1]{| #1 \rangle}

\newcommand{\ip}[2]{\langle #1|#2 \rangle}

\newcommand{\proj}[1]{| #1 \rangle \langle #1 |}

\DeclareMathOperator{\polylog}{polylog}

\newcommand{\be}{\begin{equation}}
\newcommand{\ee}{\end{equation}}
\newcommand{\bea}{\begin{eqnarray}}
\newcommand{\eea}{\end{eqnarray}}
\newcommand{\bes}{\begin{equation*}}
\newcommand{\ees}{\end{equation*}}
\newcommand{\beas}{\begin{eqnarray*}}
\newcommand{\eeas}{\end{eqnarray*}}

\newtheorem{thm}{Theorem}
\newtheorem*{thm*}{Theorem}

\newtheorem{lem}[thm]{Lemma}
\newtheorem*{lem*}{Lemma}

\begin{document}

\title{Quantum algorithms for search with wildcards and combinatorial group testing}

\author{Andris Ambainis\footnote{University of Latvia, Riga.}\ \ and Ashley Montanaro\footnote{Centre for Quantum Information and Foundations, Department of Applied Mathematics and Theoretical Physics, University of Cambridge, UK; {\tt am994@cam.ac.uk}.}}

\maketitle

\begin{abstract}
We consider two combinatorial problems. The first we call ``search with wildcards'': given an unknown $n$-bit string $x$, and the ability to check whether any subset of the bits of $x$ is equal to a provided query string, the goal is to output $x$. We give a nearly optimal $O(\sqrt{n} \log n)$ quantum query algorithm for search with wildcards, beating the classical lower bound of $\Omega(n)$ queries. Rather than using amplitude amplification or a quantum walk, our algorithm is ultimately based on the solution to a state discrimination problem. The second problem we consider is combinatorial group testing, which is the task of identifying a subset of at most $k$ special items out of a set of $n$ items, given the ability to make queries of the form ``does the set $S$ contain any special items?''\ for any subset $S$ of the $n$ items. We give a simple quantum algorithm which uses $O(k\log k)$ queries to solve this problem, as compared with the classical lower bound of $\Omega(k \log(n/k))$ queries.
\end{abstract}


\section{Introduction}

We present new quantum algorithms for two combinatorial problems. 
The first problem is {\em search with wildcards}.
In this problem, we are given an $n$-bit string $x$ and our task is to determine $x$ (so that with probability $1-\epsilon$, all bits of $x$ are correct) 
using the minimum number of queries in the following {\em wildcard query} model. 
In one wildcard query, we can check correctness of any subset of the bits of $x$.
That is, we identify queries with pairs $(S,y)$, where $S \subseteq [n]$ and $y \in \{0,1\}^{|S|}$ and the query returns 1 if $x_S = y$ 
(here the notation $x_S$ means the subset of the bits of $x$ specified by $S$). 

Wildcard queries are a generalisation of the standard quantum query model; the standard model corresponds to queries in which $S$ contains just one element. Classically, each query in this more general model still provides only one bit of information. Hence, by an information-theoretic argument classical computers still require $\Omega(n)$ queries to solve search with wildcards. Moreover, in the standard quantum query model, identifying $x$ with bounded error would require $\Omega(n)$ queries~\cite{farhi98a,beals01}. Surprisingly, in contrast to these two lower bounds, we have the following theorem.

\begin{thm}
\label{thm:sww}
There is a quantum algorithm which solves the search with wildcards problem using $O(\sqrt{n}\log n)$ queries on average\footnote{We say that an algorithm ``uses $q$ queries on average'' if the expected number of queries it makes on the worst-case input is $q$. We stress that no distribution on the inputs is assumed.}. Further, any bounded-error quantum algorithm which solves this problem must make $\Omega(\sqrt{n})$ queries.
\end{thm}

Rather than using the usual methods of designing quantum algorithms (such as amplitude amplification or quantum walks), our algorithm is based on a novel information-theoretic idea. Our algorithm gradually increases the information about the input string $x$ by repeatedly using the Pretty Good Measurement (PGM) ~\cite{belavkin75,hausladen94} to distinguish a set of quantum states. With one query, we can increase the knowledge about the input $x$ from $k$ bits to $k+\Theta(\sqrt{k})$ bits -- which
leads to a quantum algorithm using $O(\sqrt{n}\log n)$ queries. 
We think that this idea (and the natural state distinguishability problem that we solve, in Lemma \ref{lem:distinguish}),
may be of independent interest and may find more applications.

The second problem is the well known {\em combinatorial group testing} (CGT). In this problem, we are given oracle access to an $n$-bit string $x$ such that the Hamming weight of $x$ is at most $k$. We usually assume that $k$ is much smaller than $n$. In one query, we can get the OR of an arbitrary subset of the bits of $x$. The goal is to determine $x$ using the minimal expected number of queries. This models a scenario where we would like to identify a small subset of special items out of a large set of items, given the ability to make queries of the form ``does the set $S$ contain any special items?'' for any subset $S$ of the items.

The idea of combinatorial group testing\footnote{CGT is sometimes simply known as ``group testing''; we prefer the inclusion of ``combinatorial'' to avoid confusion with the notion of testing a set for being a group.} dates back to 1943, when it was proposed as a means of identifying and rejecting syphilitic men called up for induction into the US military~\cite{dorfman43}. Following this seminal work, a vast literature on the subject has developed; see the textbook~\cite{du00} for a detailed review, or the paper~\cite{porat08} for a discussion of more recent work. Areas to which efficient algorithms for CGT have been applied include molecular biology~\cite{farach97}, data streaming algorithms~\cite{cormode05}, compressed sensing~\cite{cormode06}, and pattern matching in strings~\cite{clifford10}.

Classically, it is known that the number of queries required to solve CGT is $\Theta(k \log(n/k))$~\cite{du00}. The lower bound is an information-theoretic argument while the upper bound is based on binary search. 
In the quantum case, we have the following result\footnote{A previous version of this paper claimed an upper bound of $O(\sqrt{k} \polylog(k))$ queries, via a reduction to search with wildcards. However, the reduction was incorrect and the precise quantum query complexity of CGT remains open.}.

\begin{thm}
\label{thm:main}
There is a quantum algorithm which solves the combinatorial group testing problem using $O(k \log k)$ queries on average. Further, any quantum algorithm which solves CGT with bounded error must make $\Omega(\sqrt{k})$ queries.
\end{thm}

Note that our Theorem has no dependence on $n$ (unlike the classical complexity). We prove Theorem \ref{thm:main} in two parts: a $O(k \log k)$-query quantum algorithm in Section \ref{sec:prelim} below, and a $\Omega(\sqrt{k})$ quantum lower bound in Section \ref{sec:lower}. Each part of the result is fairly straightforward.


\subsection{Related work}

One can view the search with wildcards problem as oracle interrogation -- i.e.\ learning the contents of an unknown bit-string $x$ hidden in an oracle -- in a non-standard oracle model. There has recently been some interest in this problem, in various different oracle models; we summarise the results which have been obtained as follows.

\begin{itemize}
\item First, it was shown by van Dam~\cite{vandam98} that in the standard oracle model (where the oracle performs the map $i \mapsto x_i$), there exists a quantum algorithm which learns $x$ with constant success probability using $n/2 + O(\sqrt{n})$ queries, contrasting with the $n$ classical queries required to learn $x$. Farhi et al~\cite{farhi99} later showed a matching $n/2 + \Omega(\sqrt{n})$ lower bound.

\item Iwama et al have studied the quantum query complexity of counterfeit coin problems~\cite{iwama10}. Here we are given a set of $n$ coins, $k$ of which are false (underweight), and the task is to determine the false coins. In this model, a query is specified by $q \in \{0,1,-1\}^n$ such that $\sum_i q_i = 0$. Then the oracle returns 0 if $q \cdot x = 0$, and 1 otherwise. We imagine that $x$ is a set of coins, and $x_i=0$ if the $i$'th coin is fair, and $x_i=1$ if the $i$'th coin is false. The oracle simulates a ``quantum scale'', and $q_i = 1$ (resp.\ $q_i = -1$) means that we place the $i$'th coin on the left (resp.\ right) pan. If the oracle returns 0, the scale is balanced, and if it returns 1, the scale is unbalanced. Iwama et al showed that there is a quantum algorithm based on amplitude amplification which solves this problem using only $O(k^{1/4})$ queries, beating the classical information-theoretic lower bound of $\Omega(k \log(n/k))$ queries. Note that, similarly to our algorithm for CGT, their result removes any dependence on $n$ from the complexity.

\item Finally, recently Cleve et al have studied oracle interrogation in the model of substring queries~\cite{cleve12}. Here the allowed queries are of the form ``is $y$ a substring of $x$?'' for $y \in \{0,1\}^k$, $1 \le k \le n$, where a substring of $x$ is a consecutive subsequence of $x$. Classically, this problem again requires $n$ queries; Cleve et al proved that quantum algorithms can achieve a linear speedup, giving an algorithm which uses $3n/4 + o(n)$ queries. They also show an $\Omega(n/\log^2 n)$ quantum lower bound.
\end{itemize}


\subsection{Preliminaries and notation}

We write $[n] := \{1,2,\ldots,n\}$, and use $|x|$ for the Hamming weight of $x$ and $d(x,y)$ for the Hamming distance between $x$ and $y$. For $x \in \{0,1\}^n$, a 1-index (resp.\ 0-index) of $x$ is an index $i \in [n]$ such that $x_i=1$ (resp.\ $x_i=0$). For readability, we sometimes leave states unnormalised. The two problems that we consider are precisely defined as follows:

\begin{itemize}
\item SEARCH WITH WILDCARDS.
We are given oracle access to an $n$-bit string $x$ (with no restriction on Hamming weight) and our task is to determine $x$ using the minimum number of queries. A query is specified by a string $s \in \{0,1,\ast\}^n$, and returns 1 if $x_i = s_i$ for all $i$ such that $s_i \neq \ast$, and returns 0 otherwise. We can equivalently identify queries with pairs $(S,y)$, where $S \subseteq [n]$ and $y \in \{0,1\}^{|S|}$ and the query $Q_x(S,y)$ returns 1 if $x_S = y$ (here the notation $x_S$ means the subset of the bits of $x$ specified by $S$). In the case of quantum algorithms, we give the algorithm access to the unitary oracle which maps $\ket{S}\ket{y}\ket{z} \mapsto \ket{S}\ket{y}\ket{z \oplus Q_x(S,y)}$.

\item COMBINATORIAL GROUP TESTING (CGT).
We are given oracle access to an $n$-bit string $x$ such that the Hamming weight of $x$ is at most $k$. We usually assume that $k$ is much smaller than $n$. We are allowed to query arbitrary subsets $S \subseteq [n]$ of the bits of $x$; a query $Q_x(S)$ returns 1 if there exists $i \in S$ such that $x_i = 1$. In the case of quantum algorithms, we give the algorithm access to the unitary oracle which maps $\ket{S}\ket{z} \mapsto \ket{S}\ket{z \oplus Q_x(S)}$.
\end{itemize}

We note that search with wildcards is a special case of CGT. Consider an instance of CGT where $k\le n/2$ and the input is divided into $k$ blocks $B_i = \{2i-1,2i\}$ of size 2, $1 \le i \le k$, followed by a final block of $n-2k$ bits. The input is promised to contain exactly one 1 in each of the first $k$ blocks; the position of the 1 within each block $B_i$ encodes a bit $z_i$. Now consider a subset $S$ of bits queried by an algorithm for CGT, and let $S_i = S \cap B_i$. We may assume that $S$ is a subset of the first $2k$ bits, as the last $n-2k$ bits are promised to be 0. Now observe that by choosing each $S_i$ appropriately, we can make three different kinds of query: $S_i = \{2i-1\}$ corresponds to ``does $z_i=0$?'', $S_i = \{2i\}$ corresponds to ``does $z_i=1$?'', and $S_i = \{\}$ corresponds to excluding $z_i$ from the query (the remaining query $S_i = \{2i-1,2i\}$ always returns 1 and is hence uninteresting). The overall query $S = \bigcup_i S_i$ is the OR of all of the individual queries. Thus a CGT query corresponds to a subset $S$ of the bits of $z$ and a claimed assignment $y$ to these bits; the response is 1 if any of the bits of $y$ are equal to $z$. To convert this into an instance of search with wildcards on $k$ bits, simply observe that inverting the response to such a query is equivalent to performing a query $(S,y)$ to $\bar{z}$ where the reply is 1 if $\bar{z}_S = y$. Thus an algorithm for CGT can be used to learn $\bar{z}$ and hence $z$.


\section{Search with wildcards}
\label{sec:sww}

We now show that we can indeed solve the search with wildcards problem efficiently, proving the upper bound part of Theorem \ref{thm:sww} (for the lower bound, see Section \ref{sec:lower}).
Consider an instance of search with wildcards of size $n$. Let $x\in\{0, 1\}^n$ and $k\in[n]$.

Our proof uses the following state distinguishability result (which we prove in Section \ref{sec:discrim}).
\begin{lem}
\label{lem:distinguish}
Fix $n \ge 1$ and, for any $0 \le k \le n$, set
\[ \ket{\psi^k_x} := \frac{1}{\binom{n}{k}^{1/2}} \sum_{S \subseteq [n],|S|=k} \ket{S}\ket{x_S}, \]
where $\ket{x_S} := \bigotimes_{i \in S} \ket{x_i}$. Then, for any $k = n - O(\sqrt{n})$, there is a quantum measurement (POVM) which, on input $\ket{\psi^k_x}$, outputs $\widetilde{x}$ such that the expected Hamming distance $d(x,\widetilde{x})$ is $O(1)$.
\end{lem}

In words, Lemma \ref{lem:distinguish} says that, given a superposition over $k$-subsets of the bits of $x$ with $k = n - O(\sqrt{n})$, we can output a bit-string that is likely to be very close to $x$ itself. This is in sharp contrast to the analogous situation classically; given any $n-O(\sqrt{n})$ bits of $x$, determining the remaining $O(\sqrt{n})$ bits succeeds only with exponentially small probability. Roughly speaking, our algorithm for search with wildcards will repeatedly use Lemma \ref{lem:distinguish} to learn $O(\sqrt{n})$ bits of $x$ at a time, fixing the incorrect bits after each measurement.

Consider an instance of search with wildcards of size $n$. Let $x\in\{0, 1\}^n$ and $k\in[n]$.
Recall that we denote
\[ \ket{\psi^k_x} = \sum_{S: S\subseteq [n],\\ |S|=k} \ket{S} \ket{x_S}, \]
where we write $\ket{x_S} := \otimes_{i\in S} \ket{x_i}$.
Let $M_{n, k}$ be a measurement (POVM) for distinguishing 
the states $\ket{\psi^k_x}$, and assume that $M_{n, k}$ satisfies the following property: for $k \ge n - \sqrt{n}$, and all $x$, the expected Hamming distance of the outcome $\widetilde{x}$ from $x$ is upper bounded by a constant. By Lemma \ref{lem:distinguish}, such a measurement $M_{n,k}$ indeed exists.
We can express $M_{n, k}$ as a two-step process, 
with the first step being a unitary transformation $U_{n, k}$ that maps $\ket{\psi^k_x}$
to a state in $\H_{o}\otimes \H_g$ (where $\H_o$ is the output register
and $\H_g$ is the rest of the state) and the second step being the measurement of 
$\H_{o}$ (with the measurement result interpreted as a guess $\widetilde{x}$
for the hidden bit-string $x$).

We define a sequence of numbers $n_0, \ldots, n_l$, with
$n_l=n$ and $n_{i-1}=\lceil n_i -\sqrt{n_i} \rceil$. 
Our algorithm consists of Stages 0, 1, $\ldots$, $l$.

{\bf Stage 0.} Generate $\ket{\psi^{n_0}_x}$ by first creating
$\sum_{S: S\subseteq [n],\\ |S|=n_0} \ket{S}$ and then querying each $x_i, i\in S$.

{\bf Stage $s$ ($s>0$).} Stage $s$ receives $\ket{\psi^{n_{s-1}}_x}$ as the input and
outputs $\ket{\psi^{n_s}_n}$. It consists of the following steps:
\begin{enumerate}
\item
With no queries, transform $\ket{\psi^{n_{s-1}}_x}$ to
\[ \sum_{S':S'\subseteq [n],\\ |S'|=n_s} \ket{S'} 
\sum_{S: S\subseteq S',\\ |S|=n_{s-1}} \ket{S} \ket{x_S} =
 \sum_{S:S\subseteq [n],\\ |S|=n_s} \ket{S} 
\ket{\psi^{n_{s-1}}_{x_S}} \]
\item
\label{st:subset}
Apply $U_{n_s, n_{s-1}}$ on the register holding $\ket{\psi^{n_{s-1}}_{x_S}}$. 
Use a subset query to verify whether $\widetilde{x_S}$ in the $\H_o$ register
is indeed equal to $x_S$.
Measure the outcome of the subset query.
\item
If the subset query answers that $\widetilde{x_S}=x_S$, we have a state
\[  \sum_{S:S\subseteq [n],\\ |S|=n_s} \ket{S} \ket{x_{S}} \ket{\varphi_{S}} \]
where $\ket{\varphi_{S}}$ is a state in the $\H_g$ register.
Apply the transformation $\ket{S}\ket{\varphi_{S}} \mapsto \ket{S}\ket{0}$
(which requires no queries) and discard the $\H_g$ register.
\item
\label{st:selfcorrect}
If the subset query answers that $\widetilde{x_S}\neq x_S$, repeat 
the following sequence of transformations:
\begin{enumerate}
\item
\label{st:repeat}
Use a binary search with $\lceil \log n_s \rceil$ substring queries (performed coherently,
without measurements) to find one $i$ for which $(\widetilde{x_S})_i\neq (x_S)_i$. If the algorithm succeeds, change $(\widetilde{x_S})_i$ to the opposite value.
\item
Use a subset query to verify whether $\widetilde{x_S}$ in the output register
is now equal to $x_S$. Measure the outcome of the subset query.
\item
If the subset query answers that $\widetilde{x_S}\neq x_S$, return to step \ref{st:repeat}.
\item
If the subset query answers that $\widetilde{x_S}=x_S$, we have a state
\[  \sum_{S:S\subseteq [n],\\ |S|=n_s} \ket{S} \ket{x_{S}} \ket{\varphi_{S}} \]
where $\ket{\varphi_{S}}$ is some ``garbage'' state consisting of the contents of
$\H_g$ after $U_{n_s, n_{s-1}}$ and leftover information from the subset queries
in step \ref{st:repeat}. 
Apply the transformation $\ket{S}\ket{\varphi_{S}} \mapsto \ket{S}\ket{0}$
(which requires no queries) and discard the register holding the $\ket{0}$ state.
\end{enumerate}
\end{enumerate}

The expected number of queries for Stage $s$ ($s>0$) is 1 for step \ref{st:subset} and 
$O(D \log n)$ for step \ref{st:selfcorrect}, where $D$ is the expected
number of errors in the answer $\widetilde{x_S}$. Since $D=O(1)$ by Lemma \ref{lem:distinguish}, the expected number of queries is $O(\log n)$.

For the number of stages, we can choose $l=O(\sqrt{n})$ so that $n_0= O(\sqrt{n})$.
Then, the algorithm uses $n_0=O(\sqrt{n})$ queries in Stage 0 and expected $O(\log n)$ queries
in each of the next $O(\sqrt{n})$ stages. Hence, the expected total number of queries is 
$O(\sqrt{n}\log n)$.


\section{The state discrimination problem}
\label{sec:discrim}

Our next task is to prove Lemma \ref{lem:distinguish}, i.e.\ to show that, given the state
\[ \ket{\psi^k_x} := \frac{1}{\binom{n}{k}^{1/2}} \sum_{S \subseteq [n],|S|=k} \ket{S}\ket{x_S}, \]
for any $k = n - O(\sqrt{n})$, we can output $\widetilde{x}$ such that the expected Hamming distance between $\widetilde{x}$ and $x$ is constant. We will achieve this using the pretty good measurement~\cite{belavkin75,hausladen94} (PGM), which is also known as the square root measurement~\cite{eldar01} and is defined as follows. Given a set $\{\ket{\phi_i}\}$ of pure states, set $\rho = \sum_i \proj{\phi_i}$. Then the measurement vector corresponding to state $\ket{\phi_i}$ is $\ket{\mu_i} := \rho^{-1/2} \ket{\phi_i}$, the inverse being taken on the support of $\rho$. This is a valid POVM because
\[ \sum_i \proj{\mu_i} = \sum_x \rho^{-1/2} \proj{\phi_i} \rho^{-1/2} = \rho^{-1/2} \left( \sum_i \proj{\phi_i} \right) \rho^{-1/2} = I. \]
The probability that the PGM outputs $j$ on input $\ket{\phi_i}$ is precisely $|\sqrt{G}_{ij}|^2$, where $G$ is the Gram matrix of the states $\{ \ket{\phi_i} \}$, $G_{ij} = \ip{\phi_i}{\phi_j}$. In our case, we have
\[ G_{xy} = \ip{\psi^k_x}{\psi^k_y} = \frac{1}{\binom{n}{k}} \sum_{S \subseteq [n],|S|=k} [x_S = y_S] = \frac{\binom{n-d(x,y)}{k}}{\binom{n}{k}}. \]
As $G_{xy}$ depends only on $x \oplus y$, $G$ is diagonalised by the Fourier transform over $\Z_2^n$. Eigenvalues $\lambda(s)$ of $G$, indexed by bit-strings $s \in \{0,1\}^n$, are thus given by the Fourier transform of the function $f(x) = G_{x0} = \frac{\binom{n-|x|}{k}}{\binom{n}{k}}$.
Indeed, we have
\be
\label{eq:evs} \lambda(s) = \sum_{x \in \{0,1\}^n} (-1)^{s \cdot x} f(x) = \frac{1}{\binom{n}{k}} \sum_{x \in \{0,1\}^n} (-1)^{s \cdot x} \binom{n-|x|}{k} = 2^{n-k} \frac{\binom{n-|s|}{n-k}}{\binom{n}{k}},
\ee
where the final equality is an identity of Delsarte~\cite[Eq.\ (48)]{krasikov99}.

As $\sqrt{G}_{xy}$ also depends only on $x \oplus y$, the expected Hamming distance of the output $y$ from the input $x$ does not depend on $x$ and is equal to
\[ D_k := \sum_{y \in \{0,1\}^n} d(x,y) (\sqrt{G}_{xy})^2 = \sum_{y \in \{0,1\}^n} |y| (\sqrt{G}_{0y})^2. \]
We now proceed to upper bound this quantity using Fourier duality. Observe that $D_k$ can be viewed as the inner product between the functions $f(x) = |x|$ and $g(x) = (\sqrt{G}_{0x})^2$. By Plancherel's theorem we have
\[ \sum_{x \in \{0,1\}^n} f(x) g(x) = 2^n \sum_{s \in \{0,1\}^n} \hat{f}(s) \hat{g}(s), \]
where for any function $f$ we define $\hat{f}(s) = \frac{1}{2^n} \sum_{x \in \{0,1\}^n} (-1)^{s \cdot x} f(x)$. One can easily calculate that
\[ \hat{f}(s) = \begin{cases} \frac{n}{2} & \text{if $s = 0^n$}\\ -\frac{1}{2} & \text{if $|s|=1$}\\ 0 & \text{otherwise.} \end{cases} \]
On the other hand, we can compute the Fourier spectrum of $g$ as follows. As the Fourier transform turns multiplication into convolution, we have
\[ \hat{g}(s) = \widehat{\sqrt{g}\sqrt{g}}(s) = \sum_{t \in \{0,1\}^n} \widehat{\sqrt{g}}(t) \widehat{\sqrt{g}}(s+t). \]
We can therefore determine the Fourier spectrum of $g$ directly from that of the function $\sqrt{g}(x) = \sqrt{G}_{0x}$. We have already computed this Fourier transform; up to normalisation, it is just the function giving the eigenvalues of $\sqrt{G}$, or in other words the function $\sqrt{\lambda}(s)$. We thus obtain
\beas \hat{g}(s) &=& \frac{2^{-n-k}}{\binom{n}{k}} \sum_{t \in \{0,1\}^n} \binom{n-|t|}{n-k}^{1/2} \binom{n-d(s,t)}{n-k}^{1/2}\\
&=& \frac{2^{-n-k}}{\binom{n}{k}} \sum_{t,u=0}^n |\{y:|y|=t,d(s,y)=u\}| \binom{n-t}{n-k}^{1/2} \binom{n-u}{n-k}^{1/2}.
\eeas
This is a fairly complicated expression, but as $\hat{f}(s)=0$ when $|s|>1$, we only need to calculate a few special cases. In particular, we have $\hat{g}(0^n) = 1/2^n$ and
\beas \hat{g}(e_i) &=& \frac{2^{-n-k}}{\binom{n}{k}} \sum_{t=0}^n \binom{n-t}{n-k}^{1/2} \left(\binom{n-1}{t-1}\binom{n-t+1}{n-k}^{1/2} + \binom{n-1}{t}\binom{n-t-1}{n-k}^{1/2} \right)\\
&=& 2^{-n-k} \sum_{t=0}^n \binom{k}{t} \left(\frac{t}{n} \left(\frac{n-t+1}{k-t+1} \right)^{1/2} + \left(1-\frac{t}{n}\right)\left(\frac{k-t}{n-t}\right)^{1/2} \right)\\
&=:& 2^{-n-k} \sum_{t=0}^n \binom{k}{t} T_t
\eeas
for bit-strings $e_i$ of Hamming weight 1. Thus $2^n \hat{g}(e_i)$ is equal to 1 when $k=n$ and will be close to 1 when $k$ is close to $n$. Indeed, set $k = n - c \sqrt{n}$ and consider terms $T_t$ in this sum such that $t = n/2 + a \sqrt{n}$, for $a \in \R$. Then, using the lower bound $\sqrt{x} \ge \frac{3}{2}x - \frac{1}{2}x^2$, which is valid for $x \ge 0$, we have
\beas
T_t &=& \left(\frac{1}{2} + \frac{a}{\sqrt{n}}\right)\left(1 + \frac{c}{\sqrt{n}/2 - (a+c) + 1/\sqrt{n}} \right)^{1/2} + \left(\frac{1}{2} - \frac{a}{\sqrt{n}}\right)\left(1 - \frac{c}{\sqrt{n}/2 - a} \right)^{1/2}\\
&\ge& \left(\frac{1}{2} + \frac{a}{\sqrt{n}}\right)\left(1 + \frac{c}{\sqrt{n}/2 - a} \right)^{1/2} + \left(\frac{1}{2} - \frac{a}{\sqrt{n}}\right)\left(1 - \frac{c}{\sqrt{n}/2 - a} \right)^{1/2}\\
&\ge& 1 - \frac{1}{2} \left( \frac{c}{\sqrt{n}/2-a} \right)^2 + \frac{ac}{\sqrt{n}(\sqrt{n}/2-a)}\\
&=& 1 - O(1/n)
\eeas
for constant $a$ and $c$. We thus have $2^n \hat{g}(e_i) \ge 1 - O(1/n)$. Computing the inner product $2^n \sum_{s \in \{0,1\}^n} \hat{f}(s) \hat{g}(s)$, we get
\[ D_k = \frac{n}{2}\left(1 - \hat{g}(e_i)\right) = O(1) \]
as desired. In Appendix \ref{app:discrim}, we continue the analysis of the state discrimination problem by giving quite tight upper and lower bounds on the probability of identifying $x$ exactly.


\section{Algorithms for combinatorial group testing}
\label{sec:prelim}

We now consider the related problem of combinatorial group testing, beginning by considering the very special case of CGT where $k=1$. Classically, this problem can be solved with certainty using binary search in $\lceil \log_2 n \rceil$ queries, which is asymptotically tight by an information-theoretic argument.

\begin{lem}
If $k=1$, CGT can be solved exactly using one quantum query.
\end{lem}

\begin{proof}
The result follows from observing that, in order to learn $x$, it suffices to compute the function $x \cdot s$ for arbitrary $s \in \{0,1\}^n$ (this is the same observation that underpins the quantum oracle interrogation algorithm of van Dam~\cite{vandam98}). In the CGT problem, we have access to an oracle which computes $f(s) = \bigvee_i x_i s_i$ for arbitrary $s \in \{0,1\}^n$. But if $|x| \le 1$, then for any $s$, $\bigvee_i x_i s_i = x \cdot s$. 

Formally, the quantum algorithm proceeds as follows.

\begin{enumerate}
\item Create the state $\frac{1}{\sqrt{2^{n+1}}} \sum_{s \in \{0,1\}^n} \ket{s}(\ket{0}-\ket{1})$.
\item Apply the oracle to create the state
\[ \frac{1}{\sqrt{2^{n+1}}} \sum_{s \in \{0,1\}^n} (-1)^{\bigvee_i s_i x_i} \ket{s}(\ket{0}-\ket{1}) = \frac{1}{\sqrt{2^{n+1}}} \sum_{s \in \{0,1\}^n} (-1)^{s \cdot x} \ket{s}(\ket{0}-\ket{1}) \]
\item Apply Hadamard gates to each qubit of the first register and measure to obtain $x$.
\end{enumerate}
\end{proof}

Call the above algorithm the $k=1$ algorithm. We can extend this idea to obtain a simple quantum algorithm for CGT which achieves significantly better query complexity than is possible classically, by not depending on $n$. First assume that $|x|=k$, and let $I$ be the set of 1-indices of $x$ which are currently known (initially, $I = \emptyset$). The algorithm is based on the following subroutine.

\begin{enumerate}
\item Construct a subset $S \subseteq [n]\backslash I$ by including each $i \in [n]\backslash I$ with independent probability $1/k$. Write $S_j$ for the $j$'th element of $S$.
\item Create the state $\left(\frac{1}{\sqrt{2^{|S|+1}}} \sum_{t \in \{0,1\}^{|S|}} \ket{t}\right) (\ket{0}-\ket{1})$.
\item Apply the oracle to create the state
\[ \frac{1}{\sqrt{2^{|S|+1}}} \sum_{t \in \{0,1\}^{|S|}} (-1)^{\bigvee_{i=1}^{|S|} t_i x_{S_i}} \ket{t}(\ket{0}-\ket{1}); \]
henceforth ignore the second register.
\item \label{step:hadamard} Apply Hadamard gates to each qubit of the first register to produce the state
\[ \frac{1}{2^{|S|}} \sum_{y \in \{0,1\}^{|S|}} \left(\sum_{t \in \{0,1\}^{|S|}} (-1)^{\bigvee_{i=1}^{|S|} t_i x_{S_i} + \sum_{i=1}^{|S|} t_i y_i} \right) \ket{y}. \]
\item Measure to obtain $y \in \{0,1\}^{|S|}$.
\item For each $i$ such that $y_i = 1$, add $S_i$ to $I$. Reduce $k$ by $|y|$.
\end{enumerate}

Observe that, for all $i$ such that $x_{S_i}=0$, the state produced in Step \ref{step:hadamard} has zero amplitude on all $y$ such that $y_i=1$. Thus, for each index $i$ added to $I$, $x_{S_i}=1$. On the other hand, the probability that the outcome $y=0^{|S|}$ is obtained is exactly $(1-2^{1-|x_S|})^2$, so the algorithm finds at least one 1-index with probability $1-(1-2^{1-|x_S|})^2$. In particular, if $S$ contains exactly one 1-index $i$ of $x$, which will occur with probability at least $(1-1/k)^{k-1} \ge 1/e$, we are guaranteed to learn $i$. In order to learn $x$ completely, the expected overall number of queries used is thus $O(k)$.

If we only know the upper bound that $|x|\le k$, we can simply use the above algorithm with guesses $k'=2^i$, $i=0,\dots,\lceil \log_2 k \rceil$. For at least one of these guesses $k'$, the probability that the corresponding subset $S$ contains exactly one 1-index is lower bounded by a constant, so the expected number of queries required to learn one 1-index of $x$ is $O(\log k)$. Observe that this algorithm is Las Vegas (i.e.\ it always succeeds eventually), as we can check whether we have found all 1-indices of $x$ by querying the complement of the 1-indices found so far.


\section{An almost matching lower bound}
\label{sec:lower}

We finally prove that our results for the search with wildcards and combinatorial group testing problems are almost optimal. We will use the following very general ``strong weighted adversary'' bound formulated by Zhang~\cite{zhang05} (for the statement given here, see~\cite{cleve12,spalek06}).

\begin{thm}
\label{thm:swadv}
Let $f:S \rightarrow T$ be a function and let $Q$ be a finite set of possible query strings. Let $x \in S$ be an initially unknown input which is accessed via an oracle $O_x$ performing the map $O_x \ket{q}\ket{z} = \ket{q}\ket{z \oplus \zeta(x,q)}$, where $q \in Q$, $z \in \{0,1\}$, and $\zeta:S \times Q \rightarrow \{0,1\}$ is a function specifying the response to oracle queries. Also let $w$, $w'$ be weight schemes such that:
\begin{itemize}
\item Each pair $(x,y) \in S \times S$ is assigned a non-negative weight $w(x,y)=w(y,x)$ such that $w(x,y)=0$ whenever $f(x)=f(y)$;
\item Each triple $(x,y,q) \in S \times S \times Q$ is assigned a non-negative weight $w'(x,y,q)$ such that $w'(x,y,q)=0$ for all $x$, $y$, $q$ such that $\zeta(x,q) = \zeta(y,q)$ or $f(x)=f(y)$, and $w'(x,y,q)w'(y,x,q) \ge w(x,y)^2$ for all $x$, $y$, $q$ such that $\zeta(x,q) \neq \zeta(y,q)$ and $f(x) \neq f(y)$.
\end{itemize}
For all $x \in S$ and $q \in Q$, set $\text{wt}(x) = \sum_y w(x,y)$ and $v(x,q) = \sum_y w'(x,y,q)$. Then any quantum query algorithm that computes $f(x)$ with success probability at least $2/3$ on every input $x$ must make
\[ \Omega\left( \min_{\substack{x,y,q;\,w(x,y)>0, \\ \zeta(x,q) \neq \zeta(y,q)}} \sqrt{\frac{\text{wt}(x) \text{wt}(y)}{v(x,q)v(y,q)}} \right) \]
queries to the oracle $O_x$.
\end{thm}

\begin{lem}
\label{lem:wlower}
Any quantum algorithm which solves search with wildcards on $n$ bits with worst-case success probability $2/3$ must make $\Omega(\sqrt{n})$ oracle queries.
\end{lem}

\begin{proof}
In the seach with wildcards problem the input is a string $x \in \{0,1\}^n$, queries $q=(S,t)$ are specified by $S \subseteq [n]$, $t \in \{0,1\}^{|S|}$, and $\zeta(x,q)=1$ if and only if $x_S = t$. We define the following weight scheme: $w(x,y)=1$ if $d(x,y)=1$, and $w(x,y)=0$ otherwise; $w'(x,y,q) = w'(y,x,q) = 1$ if $d(x,y)=1$ and $\zeta(x,q) \neq \zeta(y,q)$, and $w'(x,y,q) = w'(y,x,q) = 0$ otherwise. For any $x \in \{0,1\}^n$, $\text{wt}(x) = n$. On the other hand,
\[ v(x,q) = |\{y:d(x,y)=1,\zeta(x,q) \neq \zeta(y,q)\}| = \begin{cases} |S| & [\zeta(x,q)=1] \\ 1 & [\zeta(x,q)=0,\,d(x_S,t)=1] \\ 0 & [\text{otherwise}] \end{cases}. \]
Hence
\[ \min_{\substack{x,y,q;\,w(x,y)>0 \\ \zeta(x,q) \neq \zeta(y,q)}} \sqrt{\frac{\text{wt}(x) \text{wt}(y)}{v(x,q)v(y,q)}} = \sqrt{n} \]
and the claim follows from Theorem \ref{thm:swadv}.
\end{proof}

Via the reduction from search with wildcards to CGT, Lemma \ref{lem:wlower} implies that CGT requires $\Omega(\sqrt{k})$ quantum queries, completing the proof of Theorem \ref{thm:main}.


\section{Outlook}
\label{sec:outlook}

The major open question left by our work is to fully resolve the quantum query complexity of CGT. A previous version of this paper incorrectly claimed a $O(\sqrt{k} \polylog(k))$ algorithm for this problem; it is a very interesting open problem to determine its true complexity.

An alternative way of considering the CGT problem is as a restricted case of the problem of learning juntas via membership queries~\cite{mossel04,atici07}. A $k$-junta is a boolean function that depends only on at most $k$ input bits. The general problem of learning juntas is defined as follows. Given oracle access to a function $f:\{0,1\}^n \rightarrow \{0,1\}$, and the promise that $f$ is a $k$-junta, output a representation of $f$ (e.g.\ its truth table). It is easy to see that CGT is the special case of this problem where $f$ is restricted to be the OR of at most $k$ of the input bits; our algorithm therefore allows this restricted problem to be solved using $O(k\log k)$ queries. The same algorithm also works if $f$ is promised to be an AND function (i.e.\ $f(x) = \bigwedge_{i \in S} x_i$, for some $S$ such that $|S|=k$), because in this case querying $f(\bar{x})$ and negating the output simulates a query to a function $f'$ such that $f'(x) = \bigvee_{i \in S} x_i$. It would be interesting to determine whether efficient quantum algorithms could be found for other restricted cases of the junta learning problem.

A related question is {\em testing} juntas. In this problem we are given a function $f:\{0,1\}^n \rightarrow \{0,1\}$ such that $f$ either is a $k$-junta, or differs from any $k$-junta on at least $\epsilon2^n$ inputs, and must determine which is the case. Classically, this problem can be solved using $O(k/\epsilon + k \log k)$ queries~\cite{blais09}, while there is an $\Omega(k)$ lower bound on the number of queries required~\cite{chockler04}. In the quantum case, Atici and Servedio have given an $O(k/\epsilon)$-query algorithm~\cite{atici07}. It has recently been observed that there are connections between the junta testing problem and CGT~\cite{garciasoriano12}. It would be very interesting if our results could be used or generalised to give an $O(\sqrt{k}\polylog(k))$ quantum algorithm for testing juntas.


\section*{Acknowledgements}

AM was supported by an EPSRC Postdoctoral Research Fellowship and would like to thank Rapha\"el Clifford for helpful comments on a previous version, and David Gosset, Robin Kothari and Rajat Mittal for helpful discussions. 
AA was supported by ESF project 1DP/1.1.1.2.0/09/APIA/VIAA/044 and the European Commission under the project QCS (Grant No.~255961). We would like to thank Aram Harrow for suggesting a collaboration between the authors and helpful discussions. Special thanks to Han-Hsuan Lin for spotting a crucial error in a previous version.


\appendix

\section{Further analysis of the state discrimination problem}
\label{app:discrim}

In this appendix, we carry out some further analysis of the problem of discriminating the states $\ket{\psi^k_x}$ discussed in Section \ref{sec:discrim}. We have the bound from \cite{montanaro07a} that
\be \label{eq:genlower} (\sqrt{G}_{xx})^2 \ge \frac{1}{\sum_{y \in \{0,1\}^n} |\ip{\psi^k_x}{\psi^k_y}|^2 }, \ee
which allows us to prove the following lower bound on the probability that the PGM outputs $x$ exactly.

\begin{lem}
\label{lem:lower}
Set $k = n - a\sqrt{n}$ for some $0 \le a \le 1$. Then $(\sqrt{G}_{xx})^2 \ge 1 - 2 a^2 - O(1/\sqrt{n})$.
\end{lem}

\begin{proof}
By (\ref{eq:genlower}) we have
\[
(\sqrt{G}_{xx})^2 \ge \frac{\binom{n}{k}^2}{\sum_{d=0}^n \binom{n}{d}\binom{d}{k}^2} = \frac{\binom{n}{k}}{\sum_{d=0}^n \binom{d}{k}\binom{n-k}{d-k}}.
\]
We now upper bound the reciprocal of this quantity, setting $g=n-k$, $i=n-d$ to obtain
\beas
\frac{1}{\binom{n}{g}} \sum_{i=0}^g \binom{n-g+i}{i} \binom{g}{i} &=& \frac{1}{\binom{n}{g}} \sum_{i=0}^g \binom{n-g}{i} \binom{g}{i} \frac{(n-g+i)\dots(n-g+1)}{(n-g)\dots(n-g-i+1)}\\
&=& \frac{1}{\binom{n}{g}} \sum_{i=0}^g \binom{n-g}{i} \binom{g}{i} \left(1 + \frac{i}{n-g}\right) \dots \left(1 + \frac{i}{n-g-i+1} \right)\\
&\le& \frac{1}{\binom{n}{g}} \sum_{i=0}^g \binom{n-g}{i} \binom{g}{i} \left(1 + \frac{g}{n-2g+1}\right)^{g+1}\\
&\le& e^{g(g+1)/(n-2g+1)}\\
&\le& 1 + 2a^2 + O(1/\sqrt{n}).
\eeas
\end{proof}

We also record here an exact expression for the probability of getting outcome $y$ on input $x$. Let $K_k^n(x)$ be the $k$'th Krawtchouk polynomial~\cite{krasikov99}, defined by
\[ K_k^n(x) = \sum_{i=0}^k (-1)^i \binom{x}{i}\binom{n-x}{k-i}. \]
\begin{lem}
\label{lem:kexact}
\be
(\sqrt{G}_{xy})^2 = 2^{-(n+k)} \binom{n}{d(x,y)}^{-2} \left( \sum_{z=0}^n K_{d(x,y)}^n(z) \binom{n}{z}^{1/2} \binom{k}{z}^{1/2} \right)^2.
\ee
\end{lem}

\begin{proof}
Essentially immediate from the discussion in Section \ref{sec:discrim}; the entries of $\sqrt{G}$ can be calculated using
\beas
\sqrt{G}_{xy} &=& \frac{1}{2^n} \sum_{s \in \{0,1\}^n} (-1)^{(x \oplus y) \cdot s} \lambda(s)^{1/2} = \frac{1}{2^{(n+k)/2}\binom{n}{k}^{1/2}} \sum_{s \in \{0,1\}^n} (-1)^{(x \oplus y) \cdot s} \binom{n-|s|}{n-k}^{1/2}\\
&=& \frac{1}{2^{(n+k)/2}\binom{n}{k}^{1/2}} \sum_{z=0}^n \binom{n-z}{n-k}^{1/2} \sum_{s \in \{0,1\}^n,|s|=z} (-1)^{(x \oplus y) \cdot s}\\
&=& \frac{1}{2^{(n+k)/2}\binom{n}{k}^{1/2}} \sum_{z=0}^n \binom{n-z}{n-k}^{1/2} K_z^n(d(x,y)),
\eeas
where $\lambda(s)$ are the eigenvalues of $G$ (see eqn.\ (\ref{eq:evs})). Lemma \ref{lem:kexact} then follows using well-known identities for binomial coefficients and Krawtchouk polynomials~\cite{krasikov99}.
\end{proof}

We finally turn to putting upper bounds on how well $x$ can be identified given a state of the form $\ket{\psi^k_x}$. We first observe that there is no loss of generality in putting upper bounds on the success probability of the PGM, as it is in fact the optimal measurement for identifying $x$ (in terms of minimising the average probability of error). This follows from a result of Eldar and Forney~\cite{eldar01} which shows that the PGM minimises the probability of error of state discrimination for states which are geometrically uniform, i.e.\ generated by applying an abelian group to an initial state $\ket{\phi}$. This holds for our states, as they can be thought of as being generated by applying the unitary $U_y$ defined by $U_y\ket{S}\ket{x} = \ket{S}\ket{x+y_S}$ to the initial state $\sum_{S \subseteq [n],|S|=k} \ket{S}\ket{0}$. The set $\{U_y\}$, $y \in \{0,1\}^n$, clearly forms an abelian group. As a more concise proof, optimality of the PGM follows directly from the diagonal entries of $\sqrt{G}$ being equal~\cite{belavkin75}.

\begin{lem}
Set $k = n - a\sqrt{n}$ for some $a \ge 0$. Then
\[ (\sqrt{G}_{xx})^2 \le 4e^{-a^2/32}. \]
\end{lem}

\begin{proof}
We have
\[ \sqrt{G}_{xx} = 2^{-(n+k)/2} \sum_{z=0}^n \binom{n}{z}^{1/2} \binom{k}{z}^{1/2}. \]
Now split the sum into two parts to obtain
\beas \sqrt{G}_{xx} &=& 2^{-(n+k)/2} \sum_{z \le k/2 + a \sqrt{k}/4} \binom{k}{z}^{1/2} \binom{n}{z}^{1/2} + 2^{-(n+k)/2}\sum_{z > k/2 + a \sqrt{k}/4} \binom{k}{z}^{1/2}\binom{n}{z}^{1/2}\\
&\le& \left( \frac{1}{2^k}\!\!\!\!\sum_{z \le \frac{k}{2} + \frac{a \sqrt{k}}{4}} \binom{k}{z}\right)^{1/2}\!\!\left( \frac{1}{2^n}\!\!\!\!\sum_{z > \frac{k}{2} + \frac{a \sqrt{k}}{4}} \binom{n}{z}\right)^{1/2}\!\!\!\!+ \left( \frac{1}{2^k}\!\!\!\!\sum_{\frac{k}{2} + \frac{a \sqrt{k}}{4}} \binom{k}{z}\right)^{1/2}\!\!\left( \frac{1}{2^n}\!\!\!\!\sum_{z > \frac{k}{2} + \frac{a \sqrt{k}}{4}} \binom{n}{z}\right)^{1/2}
\eeas
by Cauchy-Schwarz. We now use the Chernoff bound that
\[ \frac{1}{2^n} \sum_{z \ge n/2 + b \sqrt{n}} \binom{n}{z} \le e^{-b^2/2} \]
for any $b>0$, which implies
\[ \frac{1}{2^n} \sum_{z \le k/2 + a \sqrt{k}/4} \binom{n}{z} \le e^{-a^2/32}, \;\; \frac{1}{2^k} \sum_{z > k/2 + a \sqrt{k}/4} \binom{k}{z} \le e^{-a^2/32}, \]
noting that $k/2 + a\sqrt{k}/4 \le n/2 - a\sqrt{n}/4$ by assumption. The claimed upper bound follows.
\end{proof}


\bibliographystyle{plain}
\bibliography{../thesis}

\end{document}